\documentclass{article}

\usepackage{hyperref}
\usepackage{amsmath,amssymb,amsfonts,amsthm}
\usepackage{tablefootnote}
\usepackage{xspace}
\usepackage{lettrine,color}
\usepackage{graphicx}
\usepackage{subfigure}
\usepackage{latexsym}
\usepackage{paralist}

\usepackage{algorithm}
\usepackage{algorithmicx}
\usepackage{algpseudocode}
\usepackage[T1]{fontenc}
\usepackage{lmodern}
\usepackage{thmtools}
\renewcommand{\O }{\ensuremath{{\cal O}}}

\newtheorem{theorem}{Theorem}

\newtheorem{lemma}[theorem]{Lemma}

\newtheorem{definition}[theorem]{Definition}

\newcommand{\NATURAL}{\ensuremath{\mathbb{N}}}  
 
\newcommand{\net}{\ensuremath{\mathfrak{N}}\xspace}

\newcommand{\maxdeg}{\ensuremath{\Delta}} 
\newcommand{\ecc}{\ensuremath{D}\xspace}
\newcommand{\comment}[1]{%
	\text{\phantom{(#1)}} \tag{#1}
}

\author{\textbf{Artur Czumaj} \hspace{4mm} \textbf{Peter Davies} \\[0.10in]
	Department of Computer Science \\
	Centre for Discrete Mathematics and its Applications 
	\\
	University of Warwick}

\title{\textbf{Deterministic Communication in Radio Networks}
	\thanks{Research partially supported by the Centre for Discrete Mathematics and its Applications (DIMAP).}
	\thanks{Contact information: \{A.Czumaj, P.Davies.4\}@warwick.ac.uk. Phone: +44 24 7657 3796.}
}

\begin{document}

\maketitle

\begin{abstract}
In this paper we improve the deterministic complexity of two fundamental communication primitives in the classical model of ad-hoc radio networks with unknown topology: broadcasting and wake-up. We consider an unknown radio network, in which all nodes have no prior knowledge about network topology, and know only the size of the network $n$, the maximum in-degree of any node $\maxdeg$, and the eccentricity of the network $\ecc$.

For such networks, we first give an algorithm for wake-up, based on the existence of small universal synchronizers. This algorithm runs in $O(\frac{\min\{n, \ecc \maxdeg\} \log n \log \maxdeg}{\log\log \maxdeg})$ time, the fastest known in both directed and undirected networks, improving over the previous best $O(n \log^2n)$-time result across all ranges of parameters, but particularly when maximum in-degree is small.

Next, we introduce a new combinatorial framework of block synchronizers and prove the existence of such objects of low size. Using this framework, we design a new deterministic algorithm for the fundamental problem of broadcasting, running in $O(n \log \ecc \log\log\frac{\ecc \maxdeg}{n})$ time. This is the fastest known algorithm for the problem in directed networks, improving upon the $O(n \log n \log \log n)$-time algorithm of De Marco (2010) and the $O(n \log^2\ecc)$-time algorithm due to Czumaj and Rytter (2003). It is also the first to come within a log-logarithmic factor of the $\Omega(n \log \ecc)$ lower bound due to Clementi et al.\ (2003).

Our results also have direct implications on the fastest \emph{deterministic leader election} and \emph{clock synchronization} algorithms in both directed and undirected radio networks, tasks which are commonly used as building blocks for more complex procedures.
\end{abstract}

\section{Introduction}

\subsection{Model of communication networks}

We consider the classical model of \emph{ad-hoc radio networks} with \emph{unknown structure}. A \emph{radio network} is modeled by a (\emph{directed} or \emph{undirected}) network $\net = (V,E)$, where the set of nodes corresponds to the set of transmitter-receiver stations. The nodes of the network are assigned different identifiers (IDs), and throughout this paper we assume that all IDs are distinct numbers in $\{1, \dots, |V|\}$. A directed edge $(v,u) \in E$ means that node $v$ can send a message directly to node $u$. To make propagation of information feasible, we assume that every node in $V$ is reachable in \net from any other.

In accordance with the standard model of unknown (ad-hoc) radio networks (for more elaborate discussion about the model, see, e.g., \cite{-ABLP91,-BGI92,-CGGPR00,-CGR00,-CMS03,-GHK13,-KP03,-KM98,-Pel07}), we make the assumption that a node does not have any  prior knowledge about the topology of the network, its in-degree and out-degree, or the set of its neighbors. We assume that the only knowledge of each node is its own ID, the \emph{size} of the network $n$, the \emph{maximum in-degree} of any node $\maxdeg$, and the \emph{eccentricity} of the network $\ecc$, which is the maximum distance from the source node to any node in \net.

Nodes operate in discrete, synchronous time steps, but we do not need to assume knowledge of a global clock. When we refer to the ``running time'' of an  algorithm, we mean the number of time steps which elapse before completion (i.e., we are not concerned with the number of calculations nodes perform within time steps). In each time step a node can either \emph{transmit} a message to all of its out-neighbors at once or can remain silent and \emph{listen} to the messages from its in-neighbors. Some variants of the model make restrictions upon message size (e.g. that they should be $O(\log n)$ bits in length); our algorithms only forward the source message so comply with any such restriction.

The distinguishing feature of radio networks is the interfering behavior of transmissions. In the most standard radio networks model, the \emph{model without collision detection}  (see, e.g., \cite{-ABLP91,-BGI92,-CMS03,-Pel07}), which is studied in this paper, if a node $v$ listens in a given round and precisely one of its in-neighbors transmits, then $v$ receives the message. In all other cases $v$ receives nothing; in particular, the lack of collision detection means that $v$ is unable to distinguish between zero of its in-neighbors transmitting and more than one.

The model without collision detection describes the most restrictive interfering behavior of transmissions; also considered in the literature is a less restrictive variant, the model with collision detection, where a node listening in a given round can distinguish between zero of its in-neighbors transmitting and more than one (see, e.g., \cite{-GHK13,-Pel07}).

\subsection{Discussion of assumptions of node knowledge}
\label{assump}

We consider the model that assumes that all nodes have knowledge of the parameters $n, \ecc,$ and $\maxdeg$. While these assumption may seem strong, they are standard in previous works when running time dependencies upon the parameters appear. For example, the $O(n \log^2 \ecc)$-time algorithm of \cite{-CR06} requires knowledge of $n$ and $\ecc$, and the $O(\ecc\maxdeg\log\frac n \maxdeg)$-time algorithm of \cite{-CMS03} requires knowledge of $n$ and $\maxdeg$ (though they provide methods of removing these knowledge assumptions at the expense of extra running time factors). Similar assumptions also appear in previous related work.

Furthermore, we note that nodes need only know common upper bounds for the parameters, rather than the exact values (these upper bounds will replace the true values in the running time expression). Therefore, even if only some polynomial upper bound for $\ecc$ is known, and no knowledge about $\maxdeg$ is assumed at all, our broadcasting algorithm still runs within $O(n \log \ecc \log\log\ecc)$ time, and remains the fastest known algorithm. Similarly, with only a polynomial upper bound on $\maxdeg$ and no bound on $\ecc$, our wake-up algorithm still runs in $O(\frac{n\log n \log \maxdeg}{\log\log\maxdeg})$-time. In this latter case, the algorithm is also faster than previous algorithms when only $n$ is known.

For both algorithms (as with all broadcasting and wake-up algorithms with at least linear dependency on $n$) this assumption too can be removed by standard double-and-test techniques, at the cost of never having acknowledgment of completion. The task of achieving acknowledgment in such circumstances is addressed in \cite{-V14}.

Note that to avoid non-well-defined expressions, we will use $\log (x)$ to mean $\min\{1,\log_2(x)\}$ wherever logarithms appear.

\subsection{Communications primitives: broadcasting and wake-up}

In this paper we consider two fundamental communications primitives, namely \emph{broadcasting} and \emph{wake-up}, and consider \emph{deterministic protocols} for each of these tasks.

\subsubsection{Broadcasting}

\emph{Broadcasting} is one of the most fundamental problems in communication networks and has been extensively studied for many decades (see, e.g., \cite{-Pel07} and the references therein).

The premise of the broadcasting task is that one particular node, called the \emph{source}, has a message which must become known to all other nodes. We assume that all other nodes start in a dormant state and do not participate until they are ``woken up'' by receiving the source message (this is referred to in some works as the ``no spontaneous transmissions'' rule). As a result, while the model does not assume knowledge of a global clock, we can make this assumption in practice, since the current time can be appended to the source message as it propagates, and therefore will be known be all active nodes. This is important since it allows us to synchronize node behavior into fixed-length \emph{blocks}.


\subsubsection{Wake-up}

The \emph{wake-up problem} (see, e.g.,  \cite{-JK14}) is a related fundamental communication problem that arises in networks where there is no designated ``source'' node, and no synchronized time-step at which all nodes begin communicating. The goal is for all nodes to become ``active'' by receiving some transmission. Rather than a single source node which begins active, we instead assume that some subset of nodes spontaneously become active at arbitrary time-steps. The task can be seen as broadcast from multiple sources, without the ability to assume a global clock. This last point is important, and results in wake-up protocols being slower than those for broadcast, since nodes cannot co-ordinate their behavior.


\subsection{Related work}

As a fundamental communications primitive, the task of \emph{broadcasting} has been extensively studied for various network models for many decades.

For the model studied in this paper, directed radio networks with unknown structure and without collision detection, the first sub-quadratic deterministic broadcasting algorithm was proposed by Chlebus et al.\ \cite{-CGGPR00}, who gave an $\O(n^{11/6})$-time broadcasting algorithm. After several small improvements (cf. \cite{-CGOR00,-DMP01}), Chrobak et al.\ \cite{-CGR00} designed an almost optimal algorithm that completes the task in $\O(n \log^2n)$ time, the first to be only a poly-logarithmic factor away from linear dependency. Kowalski and Pelc \cite{-KP03} improved this bound to obtain an algorithm of complexity $\O(n \log n \log\ecc)$ and Czumaj and Rytter \cite{-CR06} gave a broadcasting algorithm running in time $O(n \log^2\ecc)$. Finally, De Marco \cite{-DM10} designed an algorithm that completes broadcasting in $O(n \log n \log\log n)$ time steps. Thus, in summary, the state of the art result for deterministic broadcasting  in directed radio networks with unknown structure (without collision detection) is the complexity of $O(n\min\{\log n \log\log n,\log^2\ecc\})$ \cite{-CR06,-DM10}. The best known lower bound is $\Omega(n \log\ecc)$ due to Clementi et al.\ \cite{-CMS03}.

Broadcasting has been also studied in various related models, including undirected networks, randomized broadcasting protocols, models with collision detection, and models in which the entire network structure is known. For example, if the underlying network is undirected, then an $O(n \log \ecc)$-time algorithm due to Kowalski \cite{-K05} exists. If spontaneous transmissions are allowed and a global clock available, then deterministic broadcast can be performed in $O(n)$ time in undirected networks \cite{-CGGPR00}. Randomized broadcasting has been also extensively studied, and in a seminal paper, Bar-Yehuda et al.\ \cite{-BGI92} designed an almost optimal broadcasting algorithm achieving the running time of $\O((\ecc + \log n) \cdot \log n)$. This bound has been later improved by Czumaj and Rytter \cite{-CR06}, and independently Kowalski and Pelc \cite{-KP03b}, who gave optimal randomized broadcasting algorithms that complete the task in $O(\ecc \log \frac{n}{\ecc} + \log^2 n)$ time with high probability, matching a known lower bound from \cite{-KM98}.

Haeupler and Wajc \cite{-HW16} improved this bound for undirected networks in the model that allows spontaneous transmissions and designed an algorithm that completes broadcasting in $O(\ecc \log n \log \log n / \log\ecc + \log^{O(1)}n)$ time with high probability. In the model with collision detection for undirected networks, an $O(D+\log^6 n)$-time randomized algorithm due to Ghaffari et al. \cite{-GHK13} is the first to exploit collisions and surpass the algorithms (and lower bound) for broadcasting without collision detection.

For more details about broadcasting algorithms in various model, see e.g., \cite{-Pel07} and the references therein.

The \emph{wake-up problem} (see, e.g.,  \cite{-JK14}) is a related communication problem that arises in networks where there is no designated ``source'' node, and no synchronized time-step at which all nodes begin communicating. Before any more complex communication can take place, we must first require all nodes to be ``active,'' i.e., aware that they should be communicating. This is the goal of wake-up, and it is a fundamental starting point for most other tasks in this setting, for example leader election and clock synchronization \cite {-CGK07}.

The first sub-quadratic deterministic wake-up protocol was given in by Chrobak et al.\ \cite{-CGK07}, who introduced the concept of \emph{radio synchronizers} to abstract the essence of the problem. They give an $O(n^{5/3} \log n)$-time protocol for the wake-up problem. Since then, there have been two improvements in running time, both making use of the radio synchronizer machinery: firstly to $O(n^{3/2} \log n)$ \cite{-CK04}, and then to $O(n \log^2 n)$ \cite{-CGKR05}. Unlike for the problem of broadcast, the fastest known protocol for directed networks is also the fastest for undirected networks. Randomized wake-up has also been studied (see, e.g., \cite{-CGK07,-JS02}). A recent survey of the current state of research on the wake-up problem is given in \cite{-JK14}.


\subsection{New results}

In this paper we present a \emph{new construction of universal radio synchronizers} and introduce and analyze a \emph{new concept of block synchronizers} to improve the deterministic complexity of two fundamental communication primitives in the model of ad-hoc radio networks with unknown topology: broadcasting and wake-up.

By applying the analysis of block synchronizers, we present a new deterministic broadcasting algorithm (\textbf{Algorithm \ref{alg:BC}}) in directed ad-hoc radio networks with unknown structure, without collision detection, that for any directed network $\net$ with $n$ nodes, with eccentricity $\ecc$, and maximum in-degree $\maxdeg$, completes broadcasting in $O(n \log \ecc \log \log \frac{\ecc \maxdeg}{n})$ time-steps. This result almost matches a lower bound of $\Omega(n \, \log \ecc)$ due to Clementi et al.\ \cite{-CMS03}, and improves upon the previous fastest algorithms due to De Marco \cite{-DM10} and due to Czumaj and Rytter \cite{-CR06}, which require $O(n \log n \log\log n)$ and $O(n \log^2\ecc)$ time-steps, respectively.

Our result reveals that a non-trivial speed-up can be achieved for a broad spectrum of network parameters. Since $\maxdeg \le n$, our algorithm has the complexity at most $O(n \log \ecc \log \log \ecc)$. Therefore, in particular, it significantly improves the complexity of broadcasting for shallow networks, where $\ecc \ll n^{O(1)}$. Furthermore, the dependency on $\maxdeg$ reduces the complexity even further for networks where the product $\ecc\maxdeg$ is near linear in $n$, including sparse networks which can appear in many natural scenarios.

Our broadcasting result has also direct implications on the fastest \emph{deterministic leader election algorithm} in directed and undirected radio networks. It is known that leader election can be completed in $O(\log n)$ times broadcasting time (see, e.g., \cite{-CGR00,-GH13}) (assuming the broadcast algorithm extends to multiple sources, which is the case here as long as we have a global clock), and so our result improves the bound to achieve a deterministic leader election algorithm running in $O(n \log n \log \ecc \log \log \frac{\ecc \maxdeg}{n})$ time. For undirected networks the best result is $O(n \log^{3/2}n \sqrt{\log\log n})$ time \cite{-CKP12} (we note that the $O(n \log \ecc)$ broadcast protocol of \cite{-K05} cannot be used at a $\log n$ slowdown for leader election, since it relies on token traversal and does not extend to multiple sources). Our result therefore favorably compares for shallow networks (for small $\ecc$) even in undirected networks.

We also present a deterministic algorithm (\textbf{Algorithm \ref{alg:WU}}) for the related task of wake-up. We show the existence of universal radio synchronizers of delay $g(k) = O(\frac{n \log n \log k}{\log\log k})$, and demonstrate that this yields a wake-up protocol taking time $O(\frac{\min\{n,\ecc\maxdeg\} \log n \log \maxdeg}{\log\log \maxdeg})$. This improves over the previous best result for both directed and undirected networks, the $O(n \log^2 n)$-time protocol of \cite{-CGKR05}; the improvement is largest when $\maxdeg$ is small, but even when it is polynomial in $n$, our algorithm is a $\log\log n$-factor faster.

Our improved result for wake-up has direct applications to communication algorithms in networks that do not have access to a global clock, where wake-up is an essential starting point for most more complex communication tasks. For example, wake-up is used as a subroutine in the fastest known protocols for fundamental tasks of \emph{leader election} and \emph{clock synchronization} (cf. \cite{-CGK07}). These are two fundamental tasks in networks without global clocks, since they allow initially unsynchronized networks to be brought to a state in which synchronization can be assumed, and results from the better-understood setting with a global clock can then be applied. Our wake-up protocol yields $O(\frac{\min\{n,\ecc\maxdeg\} \log ^2n \log \maxdeg}{\log\log \maxdeg})$-time leader election and clock synchronization algorithms, which are the fastest known in both directed and undirected networks.


\subsection{Previous approaches}

Almost all deterministic broadcasting protocols with sub-quadratic complexity (that is, since \cite{-CGGPR00}) have made use of the concept of \emph{selective families} (or some similar variant thereof, such as selectors). These are families of sets for which one can guarantee that any subset of $[n]:=\{1,2,\dots,n\}$ below a certain size has an intersection of size exactly $1$ with some member of the family. They are useful in the context of radio networks because if the members of the family are interpreted to be the set of nodes which are allowed to transmit in a particular time-step, then after going through each member, any node with an active in-neighbor and an in-neighborhood smaller than the size threshold will be informed. Most of the recent improvements in broadcasting time have been due to a combination of proving smaller selective families exist, and finding more efficient ways to apply them (i.e., choosing which size of family to apply at which time).

One of the drawbacks of selective-family based algorithms is that applying them requires coordination between nodes. For the problem of broadcast, this means that some time may be wasted waiting for the current selective family to finish, and also that nodes cannot alter their behavior based on the time since they were informed, which might be desirable. For the problem of wake-up, this is even more of a difficulty; since we cannot assume a global clock, we cannot synchronize node behavior and hence cannot use selective families at all.

To tackle this issue, Chrobak et al.\ \cite{-CGK07} introduced the concept of \emph{radio synchronizers}. These are a development of selective families which allow nodes to begin their behavior at different times. A further extension to \emph{universal synchronizers} in \cite{-CK04} allowed effectiveness across all in-neighborhood sizes. However, the adaptability to different node start times comes at a cost of increased size, meaning that synchronizer-based wake-up algorithms were slightly slower than selective family-based broadcasting algorithms.

The proofs of existence for selective families and synchronizers follow similar lines: a probabilistic candidate object is generated by deciding on each element independently at random with certain carefully chosen probabilities, and then it is proven that the candidate satisfies the desired properties with positive probability, and so such an object must exist. The proofs are all non-constructive (and therefore all resulting algorithms non-explicit; cf. \cite{-I02,-CK05} for explicits construction of selective families).

Returning to the problem of broadcasting, a breakthrough came in 2010 with a paper by De Marco \cite{-DM10} which took a new approach. Rather than having all nodes synchronize their behavior, it instead had them begin their own unique pattern, starting immediately upon being informed. These behavior patterns were collated into a transmission matrix. The existence of a transition matrix with appropriate selective properties was then proven probabilistically. The ability for a node to transmit with a frequency which decayed over time allowed De Marco's method to inform nodes with a very large in-neighborhood faster, and this in turn reduced total broadcasting time from $O(n \log^2 \ecc)$ \cite{-CR06} to $O(n \log n \log\log n)$.

A downside of this new approach is that having nodes begin immediately, rather than wait until the beginning of the next selector, gives rise to a far greater number of possible starting-time scenarios that have be accounted for during the probabilistic proof. This caused the logarithmic factor in running time to be $\log n$ rather than $\log \ecc$. Furthermore, the method was comparatively slow to inform nodes of low in-degree, compared to a selective family of appropriate size. These are the difficulties that our approach overcomes.


\subsection{Overview of our approach}

Our wake-up result follows a similar line to the previous works; we prove the existence of smaller universal synchronizers than previously known, using the probabilistic method. Our improvement stems from new techniques in analysis rather than method, which allow us to gain a log-logarithmic factor by choosing what we believe are the optimal probabilities by which to construct a randomized candidate.

Our broadcasting result takes a new direction, some elements of which are new and some of which can be seen as a compromise between selective family-type objects and the transmission schedules of De Marco \cite{-DM10}. We first note that nodes of small in-degree can be quickly dealt with by repeatedly applying $(n,\frac n\ecc)$-selective families ``in the background'' of the algorithm. This allows us to tailor the more novel part of the approach to nodes of large in-degree. We have nodes performing their own behavior patterns with decaying transmission frequency over time, but they are semi-synchronized to ``blocks'' of length roughly $\frac n\ecc$, in order to cut down the number of circumstances we must consider. This idea is formalized by the concept of \emph{block synchronizers,} combinatorial objects which can be seen as an extension of the radio synchronizers used for wake-up.

An important new concept used in our analysis of block synchronizers (and also in our proof of small universal synchronizers) is that of \emph{cores}. Cores reduce a set of nodes and starting times to a (usually smaller) set of nodes which are active during a critical period. In this way we can combine many different circumstances into a single case, and demonstrate that for our purposes they all behave in the same way.

The most technically involved part of both of the proofs is the selection of the probabilities with which we generate a randomized candidate object (universal synchronizer or block synchronizer). Intuitively, when thinking about radio networks, a node in our network is aiming to inform its out-neighbors, and it should assume that as time goes on, only those with large in-neighborhoods will remain uninformed (because these nodes are harder to inform quickly). Therefore a node should transmit with ever-decreasing frequency, roughly inversely proportional to how large it estimates remaining uninformed neighbors' in-neighborhoods must be. However, the size of these in-neighborhoods cannot be estimated precisely, and so we must tweak the probabilities slightly to cover the possible range. In block synchronizers we do this using phases of length $O(\log\log\frac{\ecc\maxdeg}{n})$ during which nodes halve their transmission probability every step, but since behavior must be synchronized to achieve this we cannot do the same for radio synchronizers. Instead, we allow our estimate to be further from the true value, and require more time-steps around the same value to compensate.

As with previous results based on selective families, synchronizers, or similar combinatorial structures, the proofs of the structures we give are non-constructive, and therefore the algorithms are non-explicit.


\section{Combinatorial tools}
\label{sec:comb-tools}

Our communications protocols rely upon the existence of objects with certain combinatorial properties, and we will separate these more abstract results from their applications to radio networks. In this section, we will define the combinatorial objects we will need. Next, in Sections \ref{sec:algorithms}--\ref{sec:analysis}, we will demonstrate in detail how these combinatorial objects can be used to obtain fast algorithms for broadcasting and wake-up.


\subsection{Selective families}

We begin with a brief discussion about \emph{selective families}, whose importance in the context of broadcasting was first observed by Chlebus et al.\ \cite{-CGGPR00}. A selective family is a family of subsets of $[n] := \{1,\dots,n\}$ such that every subset of $[n]$ below a certain size has intersection of size exactly $1$ with a member of the family. For the sake of consistency with successive definitions, rather than defining the family of subsets $S_i$, we will instead use the equivalent definition of a set of binary sequences $S^v$ (that is, $S^v_i = 1$ if and only if $v \in S_i$).


For some $m \in \NATURAL$, let each $v \in [n]$ have its own length-$m$ binary sequence $S^v = S^v_0 S^v_1 S^v_2 \dots S^v_{m-1}$.

\begin{definition}
\label{def:SF}
$S= \{S^v\}_{v \in [n]}$ is an $(n,k)$-\textbf{selective family} if for any $X \subseteq [n]$ with $1 \le |X| \le k$, there exists $j$, $0 \le j < m$, such that $\sum_{v \in X}S^v_j = 1$. (We say that such $j$ \emph{hits} $X$.)
\end{definition}


\subsubsection{Existence of small selective families}

The following standard lemma (see, e.g., \cite{-CMS03}) posits the existence of $(n,k)$-selective families of size $O(k \log \frac nk)$. This has been shown to be asymptotically optimal \cite{-CMS03}.

\begin{lemma}[\textbf{Small selective families}]
\label{lem:sf}
For some constant $c$ and for any $1 \le k \le n$, there exists an $(n,k)$-selective family of size at most $m= c k \log \frac{n}{k}$.
\qed
\end{lemma}


\subsubsection{Application to radio networks}

During the course of radio network protocols we can ``apply'' a selective family $S$ on an $n$-node network by having each node $v$ transmit in time-step $j$ if and only if $v$ has a message it wishes to transmit and $S^v_j = 1$ (see, e.g., \cite{-CGGPR00,-CMS03}). Some previous protocols involved nodes starting to transmit immediately if they were informed of a message during the application of a selective family (or a variant called a selector designed for such a purpose), but here we will require nodes to wait until the current selective family is completed before they start participating. That is, nodes only attempt to transmit their message if they knew it at the beginning of the current application.

The result of applying an $(n,k)$-selective family is that any node $u$ which has between 1 and $k$ active neighbors before the application will be informed of a message upon its conclusion. This is because there must be some time-step $j$ which hits the set of $u$'s active neighbors, and therefore exactly one transmits in that time-step, so $u$ receives a message. This method of selective family application in radio networks was first used in \cite{-CGGPR00}.


\subsection{Radio synchronizers}

Radio synchronizers are an extension of selective families designed to account for nodes in a radio network starting their behavior patterns at different times, and without access to a global clock. They were first introduced in \cite{-CGK07} and used in an algorithm for performing wake-up, and this is also the purpose for which we will apply them.


To define radio synchronizers, we first define the concept of \emph{activation schedule}.

\begin{definition}
\label{def:AS}

An $n$-\textbf{activation schedule} is a function  $\omega : [n] \rightarrow \NATURAL$.
\end{definition}

We will extend the definition to subsets $X \subseteq [n]$ by setting $\omega(X) = \min_{v \in X} \omega(v)$.

As for selective families, let each $v \in [n]$ have its own length-$m$ binary sequence $S^v = S^v_0 S^v_1 S^v_2 \dots S^v_{m-1}$. We then define radio synchronizers as follows:

\begin{definition}
\label{def:radio-synchronizer}
$S= \{S^v\}_{v \in [n]}$ is called an \textbf{$(n,k,m)$-radio synchronizer} if for any activation schedule $\omega$ and for any $X \subseteq [n]$ with $1 \le |X| \le k$, there exists $j$, $\omega(X) \le j < \omega(X)+m$, such that $\sum_{v \in X} S^v_{j-\omega(v)} = 1$.
\end{definition}

One can see that the definition is very similar to that of selective families (Definition \ref{def:SF}), except that now each $v$'s sequence is offset by the value $\omega(v)$. To keep track of this shift in expressions such as the sum in the definition, we will call such values $j$ \emph{columns}. As with selective families, we say that any column $j$ satisfying the condition in Definition \ref{def:radio-synchronizer} \emph{hits $X$}.

In \cite{-CK04}, the concept of radio synchronizers was extended to universal radio synchronizers which cover the whole range of $k$ from $1$ to $n$. Let $g:[n] \rightarrow \NATURAL$ be a non-decreasing function, which we will call the \emph{delay} function. 

\begin{restatable}{definition}{durs}\label{def:urs}
$S= \{S^v\}_{v \in [n]}$ is called an $(n,g)$-\textbf{universal radio synchronizer} if for any activation schedule $\omega$, and for any $X \subseteq [n]$, there exists column $j$, $\omega(X) \le j < \omega(X)+g(|X|)$, such that $\sum_{v \in X} S^v_{j-\omega(v)} = 1$.
\end{restatable}


\subsubsection{New result: Existence of small universal radio synchronizers}
\label{sss:unsy}

We obtain a new, improved construction of universal radio synchronizers, which improves over the previous best result of Chlebus et al.\ \cite{-CGKR05} of universal synchronizers with $g(q) = O(q \log q \log n)$.

\begin{restatable}{theorem}{turs}\label{the:URS}
For any $n \in \NATURAL$, there exists an $(n,g)$-\textbf{universal radio synchronizer} with $g(q) = O(\frac{ \ q \log q \log n}{\log\log q})$.
\end{restatable}

Our approach will be to randomly generate a candidate synchronizer, and then prove that with positive probability it does indeed satisfy the required property. Then, for this to be the case, at least one such object must exist.  We will prove Theorem \ref{the:URS} in Section \ref{sss:proof-unsy}.


\subsubsection{Application of universal radio synchronizers to radio networks}

One can apply universal radio synchronizers to the problem of wake-up in radio networks by having $\omega(v)$ represent the time-step in which node $v$ becomes active during the course of a protocol (either spontaneously or by receiving a transmission). Subsequently, $v$ interprets $S^v$ as the pattern in which it should transmit, starting immediately from time-step $\omega(v)$. That is, in each time-step $j$ after activation, $v$ checks the next value in $S^v$ (i.e., $S^v_{j-\omega(v)}$), transmits if it is \textbf 1 and stays silent otherwise. Then, the selective property specified by the definition guarantees that any node $u$ with an in-neighborhood of size $q$ hears a transmission within at most $g(q)$ steps of its first in-neighbor becoming active.

We will present this approach in details in Section \ref{subsec:alg-wake-up}, where we will obtain a new, improved algorithm for the wake-up problem.


\subsection{Block synchronizers}

Next, we introduce \emph{block synchronizers}, which are a new type of combinatorial object designed for use in a fast broadcasting algorithm. They can be seen as an extension of both radio synchronizers and the transmission matrix formulation of De Marco \cite{-DM10}.


Let $\omega$ be an $n$-activation schedule (cf. Definition \ref{def:AS}). Let each $v \in [n]$ have its own length-$m$ binary sequence $S^v = S^v_0 S^v_1 S^v_2 \dots S^v_{m-1}$. For any fixed $B$, define a function $\mu_B : \NATURAL \rightarrow \NATURAL$ which rounds its input up to the next multiple of $B$, that is, $\mu_B(x) = \min\{pB : p \ge \frac{x}{B}, p \in \NATURAL\}$; we will call $s(v) := \mu_B(\omega(v))$ the \emph{start column} of $v$. We extend $s$ to subsets of $[n]$ in the obvious way, $s(X) = \mu_B(\omega(X))$.

\begin{restatable}{definition}{bs}\label{def:bs}
\label{def:block-synchronizer}
$S = \{S^v\}_{v \in [n]}$ is an $(n,\maxdeg,r,B)$-\textbf{block synchronizer} if for any activation schedule $\omega$ and any set $X \subseteq [n]$ with $ |X|\le \maxdeg$, there exists a column $j$, $s(X) \le j < s(X) + B \cdot \lceil\frac {|X|}{r}\rceil$, such that $\sum_{v \in X} S^v_{j-s(v)} = 1$.
\end{restatable}

Block synchronizers differ from radio synchronizers in two ways: Firstly, on top of the offsetting effect of the activation schedule, there is also the function $\mu_B$ that effectively~``snaps'' behavior patterns to blocks of size $B$, hence the name block synchronizer. Secondly, the size of the range in which we must hit $X$ is linearly dependent on $|X|$. This could be generalized to a generic non-decreasing function $g(|X|)$ as with universal radio synchronizers, but here for simplicity we choose to use the specific function which works best for our broadcasting application. The parameter $r$ is the increment by which each block increases the size of sets we can hit.


\subsubsection{New result: Existence of small block synchronizers}
\label{subsec:esbs}

We will show the existence of small block synchronizers in the following theorem.

\begin{restatable}{theorem}{tbs}\label{thm:bs}
For any $n, D, \maxdeg \in \NATURAL$ with $\ecc$, $\maxdeg \le n<D \maxdeg$, there exists an $(n, \maxdeg, \frac{n}{D}, O( \frac{n}{D}\log D \log\log \frac{D \maxdeg}{n}))$-\textbf{block synchronizer}.
\end{restatable}

We will prove the existence of a small block synchronizer by randomly generating a candidate $S$, and proving that it indeed has the required properties with positive probability, in a similar fashion to the proof of small radio synchronizers. We will prove Theorem \ref{thm:bs} in Section \ref{subsec:proof-esbs}.


\subsubsection{Application of block synchronizers to radio networks}

The idea of our broadcasting algorithm will be that any node $v$ waits until the start of the first block after its activation time $\omega(v)$, and then begins its transmission pattern $S^v$. The definition of block synchronizer aims to model this scenario. The hitting condition ensures that any node with an in-neighborhood of size $q \le \maxdeg$ will be informed within $B\lceil\frac{q}{r}\rceil$ time-steps of the start of the block in which its first in-neighbor begins transmitting.

We will present this approach in details in Section \ref{subsec:alg-broadcasting}, where we will obtain a new, improved algorithm for the broadcasting problem.


\section{Algorithms for broadcasting and wake-up}
\label{sec:algorithms}

In this section we use the machinery developed in the previous section to design our algorithms for broadcasting and wake-up in radio networks.


\subsection{Broadcasting}
\label{subsec:alg-broadcasting}

We will assume that $\ecc\maxdeg> n$, otherwise an earlier \sloppy{$O(\ecc\maxdeg\log\frac n\maxdeg)$-time} protocol from \cite{-CMS03} can be used to achieve $O(\ecc\maxdeg\log\frac{n}{\maxdeg}) = O(n \log \ecc)$ time.

Let $\mathcal{S}$ be an $(n,\maxdeg,\frac n \ecc,\mathcal{B})$-block synchronizer, with $\mathcal{B} = c \frac n\ecc\log \ecc \log\log \frac{\ecc \maxdeg}{n}$ (cf. Theorem \ref{thm:bs}), and recall that $\mu_{\mathcal{B}}(x) = \min\{p\mathcal{B}: p \ge \frac{x}{\mathcal{B}}, p \in \NATURAL\}$, i.e. the start of the first block after $x$. We will say that the source node becomes active at time-step $0$, and any other node $v$ becomes active in a time-step $i$ if it received its first transmission at time-step $i-1$. Our broadcasting algorithm is the following (Algorithm \ref{alg:BC}):

\begin{algorithm}[H]
\caption{Broadcast at a node $v$}
\label{alg:BC}
\begin{algorithmic}
\State Let $i$ be the time-step in which $v$ becomes active
\For {$j$ from $0$ to $D\mathcal{B}-1$, in time-step $\mu_\mathcal{B}(i)+j$}
	\State $v$ transmits source message iff $\mathcal{S}^v_{j} = 1$
\EndFor	
\end{algorithmic}
\end{algorithm}	


\subsection{Wake-up}
\label{subsec:alg-wake-up}

Let $S$ be an $(n,g)$-{universal radio synchronizer} with $g(q) = \frac{cq \log q \log n}{\log\log q}$ (cf. Theorem \ref{the:URS}). We will say that a node $v$ becomes active in a time-step $i$ if it either spontaneous wakes up at $i$, or received its first transmission at time-step $i-1$. Our wake-up algorithm is the following (Algorithm \ref{alg:WU}):

\begin{algorithm}[H]
\caption{Wake-up at a node $v$}
\label{alg:WU}
\begin{algorithmic}
\State Let $i$ be the time-step in which $v$ becomes active
\For {$j$ from $0$ to $g(n)-1$, in time-step $i+j$}
    \State$v$ transmits source message iff $S^v_{j} = 1$
\EndFor	
\end{algorithmic}
\end{algorithm}	


\section{Analysis of broadcasting and wake-up algorithms}
\label{sec:analysis}

In this section we show that our algorithms for broadcasting and wake-up have the claimed running times. Our analysis critically relies on the constructions of small block synchronizers and small universal radio synchronizers, as presented in Theorems \ref{thm:bs} and \ref{the:URS}.

We begin with the analysis of the broadcasting algorithm.

\sloppy{
\begin{theorem}
\label{the:bc}
Algorithm \ref{alg:BC} performs broadcast in $O(n\log \ecc \log\log\frac{\ecc \maxdeg}{n})$ time-steps.
\end{theorem}

To begin the analysis, fix some arbitrary node $v$ and let $P$ be a shortest path from the source (or first informed node) $x$ to $v$. Number the nodes in this path consecutively, e.g., $P_0$ = $x$ and $P_{dist(x,v)} = v$. Classify all other nodes into \emph{layers} dependent upon the furthest node along the path $P$ to which they are an in-neighbor (some nodes may not be an in-neighbor to any node in $P$; these can be discounted from the analysis). That is, $\text{layer } L_\ell = \{u\in V: \max_{u \text{ in-neighbour to }P_i}i = \ell\}$ for $\ell \le dist(x,v)$. We separately define layer $L_{dist(x,v)+1}$ to be $\{v\}$.

(For a depiction of layer numbering, see Figure \ref{figure:layers}.)
\begin{figure}[h]
\begin{center}
\includegraphics[width=0.7\textwidth]{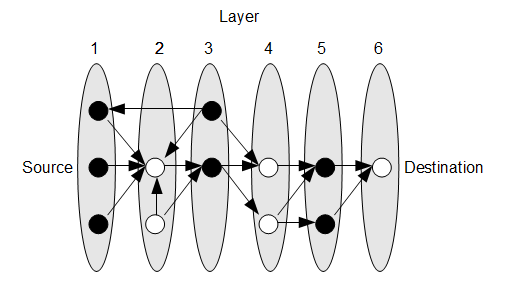}
\end{center}
\caption{An example of layer numbering.}
\label{figure:layers}
\end{figure}

At any time step, we call a layer \emph{leading} if it is the foremost layer containing an active node, and our goal is to progress through the network until the final layer is leading, i.e., $v$ is active. The use of layers allows us to restrict to the set of nodes of our main interest: if we focus on the path node whose in-neighborhood contains the leading layer, we cannot have interference from earlier layers since they contain no in-neighbors of this path node, and we cannot have interference from later layers since they are not yet active.

\begin{lemma}
\label{lem:algtime}
Let $h:[\maxdeg]\rightarrow \NATURAL$ be a non-decreasing function, and define $T(n, \ecc, \maxdeg, h)$ to be the supremum of the function $\sum_{i=1}^{\ecc} h(q_i)$, where integers $1 \le q_i \le \maxdeg$ satisfy the additional constraint $\sum_{i=1}^{\ecc} q_i \le n$. If a broadcast or wake-up protocol ensures that any layer (under any choice of $v$) of size $q$ remains leading for no more than $h(q)$ time-steps, then all nodes become active within $T(n, \ecc, \maxdeg, h)$ time-steps.
\end{lemma}

\begin{proof}
Let $q_i = |L_i|$. Layer $L_{dist(x,v)+1}$ must be leading (and thus node $v$ active) once no other layers are leading, and so this occurs within $\sum_{i=1}^{dist(x,v)}h(q_i)$ time-steps after layer $L_1$ becomes leading. Since $\sum_{i=1}^{dist(x,v)}h(q_i) \leq \sum_{i=1}^{\ecc}h(q_i)$ and $\sum_{i=1}^{\ecc} q_i \le n$, this is no more than $T(n, \ecc, \maxdeg, h)$ time-steps.

Since $v$ was chosen arbitrarily, all nodes must be active within $T(n, \ecc, \maxdeg, h)$ time-steps of $x$ becoming active.
\end{proof}

We make use of Lemma \ref{lem:algtime} to give bounds on the running times of our algorithms:

\begin{lemma}
Algorithm \ref{alg:BC} ensures that any layer of size $q$ remains leading for fewer than $\mathcal{B}\lceil\frac{q+r}{r}\rceil$ time-steps.
\end{lemma}

\begin{proof}
For all nodes $w$, let $\omega(w)$ be the time-step that $w$ becomes active during the course of the algorithm. By definition of a block selector, for any layer $L_i$ of size $q_i$ there is a time-step $j < s(L_i)+B\lceil\frac{q_i}{r}\rceil$ in which exactly one element of $L_i$ transmits. Then, either path node $P_i$ hears the transmission (and so layer $L_i$ is no longer leading in time-step $j+1$), or $P_i$ has active in-neighbors not in $L_i$, in which case these must be in a later layer so $L_i$ is not leading. Thus, $L_i$ can remain leading for no more than $s(L_i)+\mathcal{B}\lceil\frac{q_i}{r}\rceil-\omega(L_i) < \mathcal{B}\lceil\frac{q_i+r}{r}\rceil$ time-steps.
\end{proof}

With these tools, we are now ready to complete the proof of Theorem \ref{the:bc}.

\begin{proof}[\rm\textbf{Proof of Theorem \ref{the:bc}}]
By Lemma \ref{lem:algtime}, Algorithm \ref{alg:BC} ensures that all nodes are active (and have therefore heard the source message) within $T(n, \ecc, \maxdeg, h)$ time-steps, where $h(q) = \mathcal{B}\lceil\frac{q+r}{r}\rceil$. We will use an upper bound $T(n, \ecc, \maxdeg, h')$, where $h'(q) = \mathcal{B}\frac{q+2r}{r}$. Since $h'$ is linear and increasing, $\sum_{i=1}^{\ecc} h'(q_i)$ subject to $\sum_{i=1}^{\ecc} q_i \le n$ is maximized whenever $\sum_{i=1}^{\ecc} q_i = n$, for example at $q_i = \frac{n}{\ecc}$ for all $i \in [\ecc]$. So, the algorithm completes broadcast within

\begin{displaymath}
    \sum_{i=1}^{\ecc} h'(\frac{n}{\ecc})
        =
    \sum_{i=1}^{\ecc} \mathcal{B}\frac{\frac{n}{\ecc} +2r}{r}
        =
    3\mathcal{B}\ecc
        =
    3c'n\log \ecc \log\log\frac{\ecc\maxdeg}{n}
\end{displaymath}

time-steps.
\end{proof}

In a similar way, we can analyze Algorithm \ref{alg:WU}:

\begin{theorem}
\label{the:wu}
Algorithm \ref{alg:WU} performs wake-up in $O(\frac{\min(n,\ecc\maxdeg) \log n \log \maxdeg}{\log\log \maxdeg})$ time-steps.
\end{theorem}}

\begin{proof}
By Lemma \ref{lem:algtime}, and the selective property of the universal synchronizers proven in Theorem \ref{the:URS}, Algorithm \ref{alg:WU} ensures that all nodes are active within $T(n, \ecc, \maxdeg, g)$ time-steps, where $g(q) = \frac{cq \log q \log n}{\log\log q}$. Since $g$ is convex and increasing, $\sum_{i=1}^{\ecc} g(q_i)$ subject to $\sum_{i=1}^{\ecc} q_i \le n$ and $q_i \leq \maxdeg$ is maximized at $q_i = \maxdeg$ if $i \le \frac{n}{\maxdeg}$, and $q_i = 0$ otherwise. Hence, the algorithm completes wake-up within

\begin{displaymath}
    \sum_{i=1}^{\min(\ecc,\frac n\maxdeg)} g(\maxdeg)
        =
    \sum_{i=1}^{\min(\ecc,\frac n\maxdeg)} \frac{c\maxdeg \log \maxdeg \log n}{\log\log \maxdeg}
        =
    \frac{c\min(n,\ecc \maxdeg) \log n \log \maxdeg}{\log\log \maxdeg}
\end{displaymath}

time-steps.
\end{proof}


\section{Small universal radio synchronizers: Proof of Theorem \ref{the:URS}}
\label{sss:proof-unsy}

In this section we will prove our main result about the existence of small universal radio synchronizers, Theorem \ref{the:URS}. We first restate the theorem:

\turs*

Our approach will be to randomly generate a candidate synchronizer, and then prove that with positive probability it does indeed satisfy the required property. Then, for this to be the case, at least one such object must exist. We note that, since we are only concerned with asymptotic behavior, we can assume that $n$ is at least a sufficiently large constant.

Let $c$ be a constant to be chosen later. Our candidate $S= \{S^v\}_{v \in [n]}$ will be generated by independently choosing each $S^v_j$ (for $j < g(n)$) to be \textbf 1 with probability $\frac{c\log n}{6(j+c\log n)}$ and \textbf 0 otherwise.

In analyzing whether $S$ hits all sets $X \subseteq [n]$ under any activation schedule, we must first define the concept of a \emph{core} to reduce the number of possibilities we must consider.

\begin{definition}\label{def:wcore}
Fix any $X \subseteq [n]$ and any activation schedule $\omega$. Let $X_j$ be the elements of $X$ which are active by column $j$, i.e., $X_j = \{v \in X : \omega(v) \le j\}$. Let $j'$ be the smallest $j$ such that $j - \omega(X) \ge g(|X_j|)$. For every $v$, define $\psi(v) = \omega(v) -\omega(X)$, i.e., $\psi$ is $\omega$ shifted so that $\psi(X) = 0$.

The \textbf{core} $C_{X,\omega}$ of a subset $X \subseteq [n]$ with respect to activation schedule $\omega$ is defined to be

\begin{displaymath}
    \{(v,\psi(v)): \omega(v)< j'\}
\end{displaymath}
\end{definition}

This definition aims to narrow our focus to only the important elements in a particular subset $X$. Cores cut down the number of possibilities by removing redundant elements which only become active after the set must already have been hit, and by shifting activation times to begin at zero (which, as we show, can be done without loss of generality). We do not want cores to be subject to an overriding activation schedule, so we include the activation times of elements of a core within its definition. When we talk about ``hitting'' a core, we mean using these incorporated activation times rather than an activation schedule, and we assume that column numberings start at $0$ at the beginning of the core.

We note that if $S$ hits a core $C_{X,\omega}$ within $g(|C_{X,\omega}|)$ columns under $\psi$, then it hits the set $X$ within $g(|X|)$ columns under $\omega$.This result allows us to `shift' the activation times, and analyze a core independently of the many activation schedules from which it could be derived. We now need only prove that our candidate synchronizer hits all possible cores, since this will imply that it hits all subsets of $[n]$ under all activation schedules.

We make one further definition which will simplify our analysis:

\begin{definition}
For a core $C$ and column $j$, let $C(j)$ denote $\{(v,\psi(v))\in C: \psi(v) \le j\}$. The \textbf{load} of column $j$ of core $C$, denoted $f_C(j)$, is defined to be $f_C(j) = \sum_{(v,\psi(v))\in C(j)} \frac{c\log n}{6(j - \psi(v)+c\log n)}$.
\end{definition}

Note that load of a column $j$ of core $C$ is the expected number of \textbf{1}s in a column, under the probabilities used for our candidate $S$, that is, $f_C(j) = \sum_{(v,\psi(v))\in C(j)} \Pr{S^v_{j-B\phi(v)} = \mathbf{1}}$.

If $f_C(j)$ is close to constant, then the probability of $S$ hitting $C$ in column $j$ will also be almost constant. We therefore wish to bound $f_C(j)$, both from above and below.

\begin{lemma}
\label{lemma:lb-fC}
For all $j < g(|C|)$, $f_C(j) > \frac{\log\log{|C|}} {12\log{|C|}}$.
\end{lemma}

\begin{proof}
The minimum contribution each $v \in C(j)$ can add to $f_C(j)$ is $\frac{c\log n}{6(j+c\log n)}$. Hence, $f_C(j) \ge \frac{c\log n}{6(j+c\log n)} \cdot |C(j)|$. To bound this quantity, we separate into two cases:

\begin{description}
\item[Case 1: $j<c\log n$.]

In this case we can obtain an adequate bound simply using that $|C|\geq1$:

\begin{displaymath}
	\frac{c\log n}{6(j+c\log n)} \cdot |C(j)|
		\ge
	\frac{c\log n}{6(j+c\log n)}
		>
	\frac{1}{12}
		\ge
	\frac{\log\log{|C|}}{12\log{|C|}}
\end{displaymath}

\item[Case 2: $j\geq c\log n$.]
If $j < g(|C|)$, then we also have $j < g(|C(j)|)$. This can be seen by examining any set $X$ and activation schedule $\omega$ from which $C$ can be derived, and noting that
\begin{displaymath}
	j + \omega(X)
		< g(|C|) +\omega(X)
	= g(|X_{j'}|) +\omega(X)
		\le
    j'
\end{displaymath}
by Definition \ref{def:wcore}, and so
\begin{displaymath}
	j
		=
	(j+\omega(X))-\omega(X)
		<
	g(|X_{j+\omega(X)}|)
		= g(|C(j)|)
\end{displaymath}
also by Definition \ref{def:wcore}.

Recalling (cf. Theorem \ref{the:URS}) that $g(q) = \frac{cq \log q \log n}{\log\log q}$, rearranging gives $|C(j)| > \frac{j \log\log |C(j)|}{c \log n \log |C(j)|}$. Therefore total load is bounded by
\begin{displaymath}
    f_C(j)
        \ge
    \frac{c\log n}{6(j+c\log n)} \cdot |C(j)|
        >
    \frac{j\log\log |C(j)|}{6(j+c\log n)\log |C(j)|}
        \ge
    \frac{\log\log{|C|}}{12\log{|C|}}
\end{displaymath}
\end{description}
\end{proof}

This lemma provides a lower bound on $f_C(j)$. We also need an upper bound, but we cannot obtain a good one for all $j$, since transmission load in a particular column can be as large as $|C|$. We instead prove that the set of columns with load within our desired range is sufficiently large.

Let $\mathcal{F}_C = \{j <g(|C|) : \frac{\log\log{|C|}}{12\log{|C|}} < f_C(j) < \frac 12 \log\log |C|\}$. We prove the following bound:

\begin{lemma}
\label{lemma:bound-for-FC}
$|\mathcal{F}_C| \ge \frac {c |C| \log n \log |C|}{10\log\log |C|}$.
\end{lemma}

\begin{proof}
Let us first upper-bound the total load over all columns $j <g(|C|)$:
\begin{align*}
    \sum_{j <g(|C|)} f_C(j)
        &=
    \sum_{j <g(|C|)} \sum_{(v,\psi(v))\in C(j)} \frac{c\log n}{6(j - \psi(v)+c\log n)}
        \displaybreak[2]\\
        &=
    \sum_{(v,\psi(v)) \in C} \sum_{j <g(|C|)} \frac{c\log n}{6(j - \psi(v)+c\log n)}
        \\
        &\le
    \sum_{(v,\psi(v)) \in C}\int_{\psi(v)-1}^{g(|C|)-1}\frac{c\log n}{6(j - \psi(v)+c\log n)} dj
    	   \comment{\footnotesize\sf by standard integral bound}
        \\
        &=
    \frac{c\log n}{6}\sum_{(v,\psi(v)) \in C}\ln\left(\frac{g(|C|) - 1 - \psi(v)+c\log n}{c\log n-1}\right)\comment{\footnotesize\sf evaluating integral}
        \displaybreak[2]\\
        &\le
    \frac{c \log n \cdot |C|}{6} \cdot \ln\left(\frac{g(|C|)+c\log n-1}{c\log n-1}\right)
        \\
        &=
    \frac{c |C| \log n}{6} \cdot \ln\left(\frac{\frac {c |C| \log n \log |C|}{\log\log |C|}+c\log n-1}{c\log n-1}\right)\comment{\footnotesize\sf substituting $g$'s definition}
		\displaybreak[2]\\
        &\le
    \frac{c |C| \log n}{6} \cdot \ln\left(\frac{\frac {c |C| \log n \log |C|}{\log\log |C|}}{\frac12 c\log n} + 1\right)
        \displaybreak[2]\\
       	&\le
    \frac{c |C| \log n}{6} \cdot \ln(4|C|^{1.1})
        \displaybreak[2]\\
		&=\frac{1.1\ln 2 \log |C|+ \ln 4}{6} c |C| \log n
		\\
        &\le
    0.45 c |C| \log n \log |C|
\end{align*}
In the penultimate inequality we use that $\frac {2 |C| \log |C|}{\log\log |C|}+1\le 4|C|^{1.1}$, which is obvious for sufficiently large $|C|$ and can be checked manually for small $|C|$ (remembering that we consider $\log(x)$ to mean $\min\{\log_2(x),1\}$). The final inequality can be checked similarly.

Since $f_C(j) \ge 0$ for any $j < g(|C|)$, the inequality above implies that the number of columns $j < g(|C|)$ with $f_C(j) \ge \frac12 \log\log |C|$ must be fewer than $\frac{0.9 c |C| \log n \log  |C|}{\log\log |C|}$. Therefore, since by Lemma \ref{lemma:lb-fC} all elements $j \not\in \mathcal{F}_C$ must have $f_C(j) \ge \frac12 \log\log |C|$, and since $g(|C|) = \frac {c |C| \log n \log |C|}{\log\log |C|}$, we obtain:
\begin{align*}
    |\mathcal{F}_C|
        &\ge
    g(|C|) -\frac{0.9c|C| \log n \log  |C|}{\log\log |C|}
        = 
    \frac {c |C| \log n \log |C|}{10\log\log |C|}
\end{align*}
\end{proof}

Next, we will give a lower bound for the probability that $j$ hits $C$, which will later be shown to imply that columns in the set $\mathcal{F}_C$ (and hence the candidate synchronizer as a whole) have a good probability of hitting $C$. The following lemma, or variants thereof, has been used in several previous works such as \cite{-DM10}, but we prove it here for completeness.

\begin{lemma}\label{lem:hitprob}
	Let $x_i$, $i\in [n]$ be independent $\{0,1\}$-valued random variables with $\Pr{x_i=1}\leq \frac 12 \forall i$, and let $f =\sum_{i\in [n]} \Pr{x_i=1}$. Then $\Pr{\sum_{i\in [n]} x_i = 1} \geq f4^{-f}$.
\end{lemma}

\begin{proof}
	\begin{align*}
	\Pr{\sum_{i\in [n]} x_i = 1} &= \sum_{j\in [n]} \Pr{x_j = 1 \land x_i = 0 \forall i \neq j}\\
	&\geq \sum_{j\in [n]} \Pr{x_j = 1}\cdot \Pr{x_i = 0 \forall i}\\
	&\geq f\cdot \Pr{x_i = 0 \forall i}\\
	&= f\cdot \prod_{i\in [n]} (1- \Pr{x_i = 1})\\
	&\geq f\cdot \prod_{i\in [n]} 4^{-\Pr{x_i = 1}}\\
	&= f\cdot 4^{-\sum_{i\in [n]} \Pr{x_i = 1}}\\
	&=f4^{-f}
	\end{align*}
\end{proof}

For any $j$, applying this lemma with $x_v = S^v_{j-\psi(v)}$,  we get that the probability that $j$ hits $C$ is at least $f_C(j) \cdot 4^{-f_C(j)}$.

\begin{lemma}
\label{lemma:no-column-that-hits}
For any core $C$, the probability that there is no column $j<g(|C|)$ that hits $C$ is at most $1 - n^{\frac{-c |C|}{140 \ln 2}}$.
\end{lemma}

\begin{proof}
By Lemma \ref{lem:hitprob}, each column $j$ independently hits $C$ with probability at least $f_C(j) \cdot 4^{-f_C(j)}$. To proceed with the analysis we will focus on the columns in $\mathcal{F}_C$, that is, columns $j < g(|C|)$ with $\frac{\log\log{|C|}}{12\log{|C|}} < f_C(j) < \frac 12 \log\log |C|$.

Let us consider the function $1 - x 4^{-x}$ for $x>0$, and notice that this function has a global minimum at $\mu = 1/\ln 4$, is decreasing for $x<\mu$, and is increasing for $x>\mu$. For simplicity of notation, let $h$ denote the number of columns $j \in \mathcal{F}_C$ with $\mu < f_C(j) < \frac 12\log\log |C|$. Then, the probability that no columns hit is upper bounded as follows:

\begin{align*}
&\Pr{\text{no }\text{column hits}}
        \le
    \prod_{j < g(|C|)} (1-f_C(j) \cdot 4^{-f_C(j)})
        \\
        &\hspace{1cm}\le
    \prod_{j \in \mathcal{F}_C} (1-f_C(j) \cdot 4^{-f_C(j)})
        \displaybreak[2]\\
        &\hspace{1cm}=
    \prod_{\substack{j\in\mathcal{F}_C,\\\mu < f_C(j) \le \frac 12\log\log |C|}} (1-f_C(j) 4^{-f_C(j)})
    \prod_{\substack{j\in\mathcal{F}_C,\\\frac{\log\log{|C|}}{12\log{|C|}} < f_C(j) \le \mu}} (1-f_C(j) \cdot 4^{-f_C(j)})\hspace{1in}
        \\
        &\hspace{1cm}\le
    \prod_{\substack{j\in\mathcal{F}_C,\\\mu < f_C(j) \le \frac 12\log\log |C|}} \left(1-\frac{\log\log |C|}{2\log |C|}\right)
    \prod_{\substack{j\in\mathcal{F}_C,\\\frac{\log\log{|C|}}{12\log{|C|}} < f_C(j) \le \mu}} \left( 1-\frac{\log\log{|C|}}{14\log{|C|}}\right)
    	\comment{\footnotesize\sf since products are maximised by setting $f_C(j) = \frac12\log\log |C|$ and $f_C(j)=\frac{\log\log{|C|}}{12\log{|C|}}$, respectively}
        \displaybreak[2]\\
        &\hspace{1cm}\le
    \left(1-\frac{\log\log |C|}{2\log |C|}\right)^h \cdot \left( 1-\frac{\log\log{|C|}}{14\log{|C|}}\right)^{|\mathcal{F}_C|-h}
        \displaybreak[2]\\
        &\hspace{1cm}\le
    \left( 1-\frac{\log\log{|C|}}{14\log{|C|}}\right)^{|\mathcal{F}_C|}
        \\
        &\hspace{1cm}\le
    \left(1-\frac{\log\log |C|}{14\log |C|}\right)^{\frac {c |C| \log n \log |C|}{10\log\log |C|}}
        \comment{\footnotesize\sf by Lemma \ref{lemma:bound-for-FC}}
        \\
        &\hspace{1cm}\le
    e^{\frac{-c |C| \log n}{140}}
    		\comment{\footnotesize\sf using $1-x\leq e^{-x}$ for $x \in (0,1)$}
        \\
        &\hspace{1cm}=
    n^{\frac{-c |C|}{140 \ln 2}}
\end{align*}
\end{proof}

We now have a lower bound on the probability that $S$ hits a particular core, but it remains to bound the number of possible cores we must hit.

Let $C_q$ be the set of possible cores of size $q$.

\begin{lemma}
\label{lemma:bound-for-Cq}
$|C_q| \le n^{3q}$.
\end{lemma}

\begin{proof}
There are at most $n \cdot g(n)$ possible pairs of $(v,\psi(v))$, and thus at most $\binom{n \cdot g(n)}{q}$ ways of choosing a size-$q$ subset. So, $|C_q|$ is at most $\binom{n \cdot g(n)}{q} \le (n \cdot g(n))^q = (\frac{c n^2 \log^2 n}{\log\log n})^q \le n^{3q}$ (for sufficiently large $n$).
\end{proof}

We are now ready to prove our existence result:

\begin{lemma}
With positive probability, $\mathcal{S}$ is an $(n,g)$-universal synchronizer.
\end{lemma}

\begin{proof}
We will set $c$ to be $700 \ln 2$. By union bound, using Lemmas \ref{lemma:no-column-that-hits} and \ref{lemma:bound-for-Cq},
\sloppy{
\begin{align*}
    &\Pr{\mathcal{S}\text{ is an $(n,g)$-universal synchronizer}}
        \le
    \sum_{q=1}^{n}\sum_{C \in C_q} \Pr{\text{C is not hit}}\\
        &\hspace{1cm}\le
    \sum_{q=1}^{n}\sum_{C \in C_q} n^{\frac{- c |C|}{140 \ln 2}}
        \le
    \sum_{q=1}^{n} n^{3q} \cdot n^{\frac{- c q}{140 \ln 2}}
        =
    \sum_{q=1}^{n} n^{(3-\frac{c}{140 \ln 2}) q}\\
        &\hspace{1cm}\le
    \sum_{q=1}^{n} n^{-2q}
        <
    1\enspace.
\end{align*}
}
\end{proof}

We are now ready to prove Theorem \ref{the:URS}:

\begin{proof}
Since our candidate $S$ satisfies the properties of an $(n,g)$-universal radio synchronizer with positive probability, such an object must exist. This completes the proof of Theorem \ref{the:URS}.
\end{proof}


\section{Small block synchronizers: Proof of Theorem \ref{thm:bs}}
\label{subsec:proof-esbs}

In this section we will prove our main result about the existence of small block synchronizers, Theorem \ref{thm:bs}. We first restate the theorem:

\tbs*

As in our proof of the existence of small radio synchronizers (see Section \ref{sss:proof-unsy}), we only consider the case where $n$ is at least a sufficiently large constant, since we are only concerned with asymptotic behavior. We will again need to define the \emph{core} of a subset of $[n]$ (with respect to an activation schedule $\omega$) in order to reduce the amount of possible circumstances we will consider. The main difference to our definition of cores in Section \ref{sss:proof-unsy} is that we need only retain the relative values of $\omega$ to the nearest \emph{block}, rather than keeping the exact (shifted) values. This is the reason for us introducing the concept of blocks (and block synchronizers), and it allows the range of possible cores to be cut down substantially.

\begin{definition}\label{def:bcore}
Fix any $X \subseteq [n]$ and activation schedule $\omega$. Let $X_j$ be the elements of $X$ which are active by the start of the block containing column $j$, i.e., $X_j = \{v \in X : s(v) \le j\}$. Let $j'$ be the smallest $j$ such that $j - s(X) \ge \frac {B \cdot |X_j|}{r}$.

For every $v$, define $\phi(v) = \frac{s(v) -s(X)}{B} $, i.e., $\phi(v)$ is the number of blocks that pass between the start column of $X$ and the start column of $v$. Note that $\phi(v) \in \NATURAL$.

The \textbf{core} $\mathcal{C}_{X,\omega}$ of a subset $X \subseteq [n]$ with respect to activation schedule $\omega$ is defined to be

\begin{displaymath}
    \{(v,\phi(v)): v \in X, s(v)< j'\}
\end{displaymath}
\end{definition}

We see, as we did in Section \ref{sss:proof-unsy}, that if some object $S$ ``hits'' all cores, then it hits all subsets of $[n]$ under any activation schedule. By hitting a core $\mathcal{C}$ at column $j$, we mean that $\sum_{(v,\phi(v)) \in \mathcal{C}}S^v_{j-B\phi(v)} = 1$, and we assume column numberings start at the beginning of the core. So, if $S$ hits a core $\mathcal{C}_{X,\omega}$ within $\frac {B \cdot |\mathcal{C}_{X,\omega}|}{r}$ columns, then it hits the set $X$ within $\frac {B \cdot |X|}{r}$ columns of $s(X)$ under activation schedule $\omega$.

We wish to prove the existence of a small block synchronizer by randomly generating a candidate $S$, and proving that it indeed has the required properties with positive probability, in a similar fashion to the proof of small radio synchronizers. While this could be achieved directly, we can in fact get a better result by proving existence of a slightly weaker object using this method, and then bridging the gap with selective families.

\begin{definition}
$S = \{S^v\}_{v \in [n]}$ is an $(n,k,\maxdeg,r,B)$-\textbf{upper block synchronizer} if, for any core $\mathcal{C}$ with $k \le |\mathcal{C}|\le \maxdeg$, there exists column $j< \frac {B \cdot |\mathcal{C}|}{r}$  such that $\sum_{(v,\phi(v)) \in\mathcal{C}} S^v_{j-B\phi(v)} = 1$.
\end{definition}

An upper block synchronizer has a lower bound $k$ on the size of the cores it must hit. To obtain our full block synchronizer result, we will first show the existence of small upper block synchronizers, and then show that these can be extended to block synchronizers by adding selective families to hit cores of size less than $k$.

\begin{theorem}
\label{thm:partial-ups}
For some constant $c$ and for any $n,D,\maxdeg$ with $D,\maxdeg \le n < D\maxdeg$, there exists an \sloppy{ $(n,\frac nD,\maxdeg,\frac nD,c \frac nD\log D \log\log \frac{ D \maxdeg}{n})$-upper block synchronizer.}
\end{theorem}

\begin{proof}
Let $c$ be a constant to be chosen later. For simplicity of notation we now set $k = \frac nD$, $r = \frac nD$, and $B = c \frac nD\log D \log\log \frac{D \maxdeg}{n}$.

Define $\rho(j) = j \bmod 2\log\log\frac{D\maxdeg}{n}$. Our candidate upper block synchronizer $S= \{S^v\}_{v \in [n]}$ will be generated by independently choosing each $S^v_j$ (for $j < \frac{nB}{r}$) to be \textbf 1 with probability $\frac{c\log D \log\log \frac{D\maxdeg}{n}}{(B+j)2^{\rho(j)+1}}$ and \textbf 0 otherwise.

We will analyze our candidate upper block synchronizer by fixing some particular core and bounding the probability that the candidate hits it. We begin by defining the \emph{load} of a column (with respect to some fixed core $\mathcal{C}$), and bounding it both above and below on a subset of columns. As before, load represents expected number of \textbf{1}s in a column, and we want it to be constant in order to maximize hitting probability. Recall that we now consider column numbering to begin at the start of the core, i.e. $\min_{(v,\phi(v)) \in \mathcal{C}}\phi(v) = 0$.

\begin{definition}
Let $\mathcal{C}(j)$ denote $\{(v,\phi(v))\in \mathcal{C}: B\phi(v) \le j\}$. The \textbf{load} of a column $j$ of core $\mathcal{C}$, denoted $f_\mathcal{C}(j)$, is defined to be $\sum_{(v,\phi(v))\in \mathcal{C}(j)} \Pr{S^v_{j-B\phi(v)} = \mathbf{1}} = \sum_{(v,\phi(v))\in \mathcal{C}(j)}\frac{c\log D \log\log \frac{D\maxdeg}{n}}{(j-B\phi(v) + B)2^{\rho(j)+1}}$.
\end{definition}

Since load varies across a wide range during each $2\log\log \frac{D\maxdeg}{n}$-length ``phase,'' we first consider only the columns at the start of each phase (i.e., those $j$ with $\rho(j) = 0$), which we will call \emph{0-columns}.

\begin{lemma}
\label{lemma:lb-fc}
For all $\frac B2 \le j < \frac {B \cdot |\mathcal{C}|}{r}$ with $\rho(j) = 0$, $f_\mathcal{C}(j) > \frac{1}{6}$.
\end{lemma}

\begin{proof}
Recall that, when deriving a core from a set $X$, we ended the core at the first column $j'$ with $ j' - s(X) \ge \frac {B \cdot |X_j|}{r}$, i.e. for all $j\leq j'-1$, $j-s(X)< \frac {B \cdot |X_j|}{r}$. Having shifted column numberings, this implies that for $j<\frac {B \cdot |\mathcal{C}|}{r}$, $j < \frac {B \cdot |\mathcal{C}(j)|}{r}$. The minimum contribution any$(v,\phi(v)) \in \mathcal{C}(j)$ can add to $f_\mathcal{C}(j)$ is $\frac{c\log D \log\log \frac{D\maxdeg}{n}}{2(j + B)}$. Therefore total load is upper bounded by
\begin{displaymath}
    f_\mathcal{C}(j)
        \ge
    |\mathcal{C}(j)| \cdot \frac{c\log D \log\log \frac{D\maxdeg}{n}}{2(j + B)}
        >
    \frac{cj}{2c(j+B)}
        \ge
    \frac16
\end{displaymath}
\end{proof}

This lemma provides a lower bound on $f_\mathcal{C}(j)$. We also need an upper bound, but we cannot obtain a good one for all $j$, since load in a particular column can be very large. We circumvent this issue by only bounding the load on a smaller set of columns.

Let $\mathcal{F_C} = \{j <\frac {B \cdot |\mathcal{C}|}{r} : \rho(j) = 0, \frac 16 < f_\mathcal{C}(j) < 3\log\frac{|\mathcal{C}|D}{n}\}$. We prove a lower bound on $|\mathcal{F_C}|$.

\begin{lemma}
\label{lemma:lb-for-0-column-FC}
If $\frac{n}{D} \le |\mathcal{C}| \le \maxdeg$, then $|\mathcal{F}_C| \ge \frac {c}{6} |\mathcal{C}|\log D$.
\end{lemma}

\begin{proof}
We first upper bound the total load of all 0-columns $j$ with $j <\frac {B \cdot |\mathcal{C}|}{r}$ and then show that not too many of these columns can have $f_\mathcal{C}(j) \ge 3\log\frac{|\mathcal{C}|D}{n}$, giving a lower bound for the number of 0-columns in $\mathcal{F_C}$.

We bound the total load of all 0-columns $j$ with $j <\frac {B \cdot |\mathcal{C}|}{r}$ as follows:
\sloppy{
\begin{align*}\allowdisplaybreaks
    \sum_{\substack{j < \frac{B \cdot |\mathcal{C}|}{r}\\ \rho(j)=0}} f_\mathcal{C}(j)&
        =
    \sum_{\substack{j <\frac {B \cdot |\mathcal{C}|}{r}\\ \rho(j)=0}} \sum_{(v,\phi(v))\in \mathcal{C}(j)}\frac{c\log D \log\log \frac{D\maxdeg}{n}}{2(j-B\phi(v) + B)}
        \displaybreak[2]\\
        &=
    \sum_{(v,\phi(v))\in \mathcal{C}}\sum_{\substack{B\phi(v)\le j <\frac {B \cdot |\mathcal{C}|}{r}\\ \rho(j) = 0}}\frac{c\log D \log\log \frac{D\maxdeg}{n}}{2(j-B\phi(v) + B)}
        \displaybreak[2]\\
		&=
    \sum_{(v,\phi(v))\in \mathcal{C}}\sum\limits_{i=\frac{B\phi(v)}{2\log\log\frac{D\maxdeg}{n}}}^{\frac {B \cdot |\mathcal{C}|}{2r\log\log\frac{D\maxdeg}{n}}-1}\frac{c\log D \log\log \frac{D\maxdeg}{n}}{2(2i\log\log \frac{D\maxdeg}{n}-B\phi(v) + B)}
	    \comment{\footnotesize\sf substitution of sum index variable}
        \displaybreak[2]\\
        &\le
    \sum_{(v,\phi(v))\in \mathcal{C}}\int_{\frac{B\phi(v)}{2\log\log\frac{D\maxdeg}{n}}-1}^{\frac {B \cdot |\mathcal{C}|}{2r\log\log\frac{D\maxdeg}{n}}-1}\frac{c\log D \log\log \frac{D\maxdeg}{n}}{2(2i\log\log \frac{D\maxdeg}{n}-B\phi(v) + B)} di
    		\comment{\footnotesize\sf using standard integral bound}
        \displaybreak[2]\\
        &=
    \frac{c\log D}{4} \sum_{(v,\phi(v))\in \mathcal{C}}\ln \left(\frac{\frac {B \cdot |\mathcal{C}|}{r}-2\log\log\frac{D\maxdeg}{n}-B\phi(v) + B}{B-2\log\log\frac{D\maxdeg}{n}}\right)
    		\comment{\footnotesize\sf evaluating integral}
        \displaybreak[2]\\
        &\le
    \frac{c|\mathcal{C}|\log D}{4} \ln \left(\frac{\frac {B \cdot |\mathcal{C}|}{r}-2\log\log\frac{D\maxdeg}{n} + B}{B-2\log\log\frac{D\maxdeg}{n}}\right)
        \displaybreak[2]\\
        &=
    \frac{c|\mathcal{C}|\log D}{4}
    \ln \left(
        \frac{|\mathcal{C}| c \log D \log\log\frac{D\maxdeg}{n} - 2\log\log\frac{D\maxdeg}{n} + B}
        {B-2\log\log\frac{D\maxdeg}{n}}
    \right)
        \displaybreak[2]\\
        &\le
    \frac{c|\mathcal{C}|\log D}{4} \ln \left(\frac{2(c|\mathcal{C}|\log D \log\log \frac{D\maxdeg}{n} +B)}{B}\right)
        \displaybreak[2]\\
        &=
    \frac{c|\mathcal{C}|\log D}{4} \ln \left(\frac{2(|\mathcal{C}| +\frac nD)}{\frac nD}\right)
        \\
        &\le
    \frac14 c|\mathcal{C}|\log D \ln \frac{4|\mathcal{C}|D}{n}
            \comment{\footnotesize\sf using the assumption $\frac{n}{D} \le |\mathcal{C}|$}
        \displaybreak[2]\\
        &\le
    \frac14 c|\mathcal{C}| \log D \log \frac{|\mathcal{C}|D}{n}
\end{align*}}

Since for any $j < \frac{B \cdot |\mathcal{C}|}{r}$ we have $f_\mathcal{C}(j) > 0$, the inequality above implies that there must be not more than $\frac{1}{12} c |\mathcal{C}|\log D$ 0-columns with $f_\mathcal{C}(j) \ge 3 \log \frac{|\mathcal{C}|D}{n}$.
By Lemma \ref{lemma:lb-fc}, the number of columns $j$ with $j < \frac{B \cdot |\mathcal{C}|}{r}$ for which $f_\mathcal{C}(j) \le \frac{1}{6}$ is at most $\frac{B}{2}$, and hence the number of such 0-columns is at most $\frac{B}{4\log\log\frac{D\maxdeg}{n}}$.
Therefore, $|\mathcal{F_C}|$, which is the number of 0-columns $j$ with $j < \frac{B \cdot |\mathcal{C}|}{r}$ for which $\frac16 < f_\mathcal{C}(j) < 3 \log \frac{|\mathcal{C}|D}{n}$, is upper bounded as follows:

\begin{align*}
    |\mathcal{F_C}|
        &\ge
    \frac{B \cdot |\mathcal{C}|}{2 r \log\log\frac{D\maxdeg}{n}}
        -
    \frac{B}{4\log\log\frac{D\maxdeg}{n}} - \frac{1}{12} c |\mathcal{C}| \log D\\
        &=
    \frac{c}{2} \log D
            \left(|\mathcal{C}|-\frac{n}{2D}-\frac{|\mathcal{C}|}{6}\right)
        \ge
    \frac{c}{6} |\mathcal{C}| \log D
\end{align*}

where the last inequality follows from our assumption that $\frac{n}{D} \le |\mathcal{C}|$.
\end{proof}

With the bound of the load of 0-columns in Lemma \ref{lemma:lb-for-0-column-FC}, we can obtain a significantly tighter bound on a subset of all columns.

Let $\mathbb{F}_\mathcal{C} = \{j<\frac {B \cdot |\mathcal{C}|}{r}: \frac 16 < f_\mathcal{C}(j) \le 2\}$.

\begin{lemma}
\label{lem:fi}
For any $\mathcal{C}$ with $\frac nD \le |\mathcal{C}| \le \maxdeg$,  $|\mathbb{F}_\mathcal{C}| \ge\frac {c}{12} |\mathcal{C}|\log D$.
\end{lemma}

\begin{proof}
We show that, whenever we have a $0$-column with load in the range $(\frac 16 , 3\log\frac{|\mathcal{C}|D}{n})$, there must be some column within the same phase for which load is in the range $(\frac 16, 2)$.

For any $j \in \mathcal{F}_\mathcal{C}$, let $j' = j + \log f_\mathcal{C}(j) - 1$. Then,
\begin{displaymath}
    j'
        <
    j + \log (3\log\frac{|\mathcal{C}|D}{n}) - 1
        <
    j+2\log\log\frac{D \maxdeg}{n}
\end{displaymath}

so $j'$ is in the same phase as $j$ (i.e., $j-\rho(j) = j'-\rho(j')$). Hence,

\begin{align*}
    f_\mathcal{C}(j')
        &=
    \sum_{(v,\phi(v))\in \mathcal{C}(j')}\frac{c\log D \log\log \frac{D\maxdeg}{n}}{(j'-B\phi(v) + B)2^{\rho(j')+1}}
        \\
        &=
    \sum_{(v,\phi(v))\in \mathcal{C}(j)}\frac{c\log D \log\log \frac{D\maxdeg}{n}}{(j'-B\phi(v) + B)2^{\rho(j)+\log f_\mathcal{C}(j)}}
        \\
        &=
    \sum_{(v,\phi(v))\in \mathcal{C}(j)}\frac{c\log D \log\log \frac{D\maxdeg}{n}}{(j'-B\phi(v) + B)f_\mathcal{C}(j)}
        \\
        &=
    \frac{2}{f_\mathcal{C}(j)}\sum_{(v,\phi(v))\in \mathcal{C}(j)}\frac{c\log D \log\log \frac{D\maxdeg}{n}}{2(j-B\phi(v) + B)}\cdot \frac{(j-B\phi(v) + B)}{(j'-B\phi(v) + B)}
\end{align*}

Since, for any $(v,\phi(v))\in \mathcal{C}(j)$, $\frac 13< \frac{1}{1+\frac{2\log\log\frac{D \maxdeg}{n}}{B}} \le \frac{(j-B\phi(v) + B)}{(j'-B\phi(v) + B)}\le 1$, we can bound $f_\mathcal{C}(j')$ from above:

\begin{displaymath}
    f_\mathcal{C}(j')
        \le
    \frac{2}{f_\mathcal{C}(j)}\sum_{(v,\phi(v))\in \mathcal{C}(j)}\frac{c\log D \log\log \frac{D\maxdeg}{n}}{2(j-B\phi(v) + B)}\cdot 1
        = 2
\end{displaymath}

and below:

\begin{displaymath}
   	f_\mathcal{C}(j')
		>
	\frac{2}{f_\mathcal{C}(j)}\sum_{(v,\phi(v))\in \mathcal{C}(j)}\frac{c\log D \log\log \frac{D\maxdeg}{n}}{2(j-B\phi(v) + B)}\cdot \frac 13
		=
    \frac23
\end{displaymath}

(The reason we allow loads to be as low as $\frac16$ in the definition of $\mathbb{F}_\mathcal{C}$ is to account for cases where $f_\mathcal{C}(j)\leq 2$ and so $j' = j$.)

Therefore $j' \in \mathbb{F}_\mathcal{C}$. This mapping of $j$ to $j'$ is an injection from $\mathcal{F}_\mathcal{C}$ to $\mathbb{F}_\mathcal{C}$, and so $|\mathbb{F}_\mathcal{C}|\ge |\mathcal{F}_\mathcal{C}| \ge\frac {c}{12} |\mathcal{C}|\log D$.
\end{proof}

Now that we have proven that sufficiently many columns have loads within a constant-size range, we want to show that $S$ has a good probability of hitting $\mathcal{C}$ on these columns. To do so, we again apply Lemma \ref{lem:hitprob}, setting $x_v = \mathcal{S}^v_{j-B\phi(v)}$, and see that the probability of $S$ hitting $\mathcal{C}$ on column $j$ is at least $f_\mathcal{C}(j) \cdot 4^{-f_\mathcal{C}(j)}$

\begin{lemma}
For any core $\mathcal{C}$ with $\frac nD \le |\mathcal{C}|\le \maxdeg$, with probability at least $1 - D^{-\frac{c |\mathcal{C}|}{63}}$ there is a column $j < \frac{B \cdot |\mathcal{C}|}{r}$ on which $S$ hits $\mathcal{C}$.
\end{lemma}

\begin{proof}
Let us first recall that $\mathbb{F}_\mathcal{C} = \{j < \frac {B \cdot |\mathcal{C}|}{r}: \frac16 < f_\mathcal{C}(j) \le 2\}$, and note that function $h(x) = 1 - x 4^{-x}$ for $\frac16 \le x \le 2$ is maximized at $x = 2$, with $h(2) = \frac 78$.

Each column $j$ independently hits $C$ with probability at least $f_\mathcal{C}(j) \cdot 4^{-f_\mathcal{C}(j)}$, so the probability that none hit is bounded by:

\begin{align*}
    \Pr{\text{no column hits}}
        &\le
    \prod_{j<\frac {B \cdot |\mathcal{C}|}{r}} (1-f_\mathcal{C}(j) \cdot 4^{-f_\mathcal{C}(j)})
        \le
    \prod_{j \in \mathbb{F}_\mathcal{C}} (1-f_\mathcal{C}(j) \cdot 4^{-f_\mathcal{C}(j)})
        \\
        &\le
    \prod_{j \in \mathbb{F}_\mathcal{C}} \frac 78
        \le
    \left(\frac 78\right)^{\frac {c}{12} |\mathcal{C}|\log D }
        =
    D^{-\frac{c}{12} |\mathcal{C}|\log \frac 78 }
        \le
    D^{-\frac{c \cdot |\mathcal{C}|}{63}}
\end{align*}

where the penultimate inequality follows from Lemma \ref{lem:fi}.
\end{proof}

We have a bound on the probability of hitting a particular core, but before we can show that we can hit all of them, we must count the number of possible cores.

Let $\mathcal{C}_q$ be the set of possible cores of size $q$.

\begin{lemma}
$|\mathcal{C}_q| \le D^{2q}$.
\end{lemma}

\begin{proof}
For any $(v,\phi(v))\in \mathcal{C}$, $B\phi(v)<\frac{B|C|}{r}$, i.e., for a core of size $q$, $\phi(v)<\frac qr$. Therefore there are at most $n \cdot \frac{q}{r}$ possible pairs of $(v,\phi(v))$, and thus at most $\binom{n \cdot \frac{q}{r}}{q}$ ways of choosing a size-$q$ subset. So, $|\mathcal{C}_q|$ is at most $\binom{n q/r}{q} =\binom{Dq}{q} \le \left(eD\right)^q \le D^{2q}$.
\end{proof}

We are now ready to prove the existence of a small upper block synchronizer:

\sloppy{\begin{lemma}
With positive probability, $S$ is an $(n,\frac nD,\maxdeg,\frac nD,c \frac nD\log D \log\log \frac{ D \maxdeg}{n})$-upper block synchronizer.
\end{lemma}}

\begin{proof}
We will set $c$ to be 189. By union bound,
\begin{align*}
    &\Pr{S \text{ is not an upper block synchronizer}}
        \le
    \sum_{q=\frac nD}^{\maxdeg}\sum_{C \in C_q} \Pr{\text{C is not hit}}\\
        &\hspace{1cm}\le
    \sum_{q=\frac nD}^{\maxdeg}\sum_{C \in C_q} D^{-cq/63}
        \le
    \sum_{q=\frac nD}^{\maxdeg} D^{2q} D^{-cq/63}
        =
    \sum_{q=\frac nD}^{\maxdeg} D^{2q} D^{-3q}\\
        &\hspace{1cm}=
    \sum_{q=\frac nD}^{\maxdeg} D^{- q}
        <
    \frac2D
        <
    1
\end{align*}
\end{proof}

Since, with positive probability, our candidate $S$ is an $(n,\frac nD,\maxdeg,\frac nD,c \frac nD\log D \log\log \frac{ D \maxdeg}{n})$-upper block synchronizer, at least one such object must exist, and so we have completed our proof of Theorem \ref{thm:partial-ups}.
\end{proof}

We can now prove Theorem \ref{thm:bs}:

\begin{proof}
We construct block synchronizer $\mathcal{S}$ by taking an $(n,\frac nD,\maxdeg,\frac nD,c \frac nD\log D \log\log \frac{ D \maxdeg}{n})$-upper block synchronizer $S$ and inserting an $(n,\frac nD)$-selective family $R$ of size $\tilde{c}\frac nD\log D \log\log \frac{D \maxdeg}{n}$ at the beginning of each block (we know by Lemma \ref{lem:sf} that a selective family of size $\tilde{c}\frac nD\log D$ exists, and we can pad it arbitrarily to this larger size). That is, our block size will now be $\mathcal{B} := |R|+B = (c+\tilde{c})\frac nD\log D \log\log \frac{D \maxdeg}{n}$, and our block synchronizer $\mathcal{S}$ will be formally defined by:
\begin{displaymath}
    \mathcal{S}
        =
    \{\mathcal{S}^v\}_{v \in [n]} \text{ is defined by } \mathcal{S}^v_j
        =
    \begin{cases}
        R^v_{j \bmod \mathcal{B}} &\text{if } (j \bmod \mathcal{B}) < |R|,\\
        S^v_{j - \lceil\frac{j}{\mathcal{B}}\rceil R} &\text {otherwise.}
    \end{cases}
\end{displaymath}

Setting $\hat{c} = c+\tilde{c}$, we show that $\mathcal{S}$ satisfies the conditions of an $(n, \maxdeg, \frac{n}{D}, \hat{c} \frac{n}{D}\log D \log\log \frac{D \maxdeg}{n})$-block synchronizer.

Let $\mathcal{C}$ be a core of size at most $\maxdeg$.

\begin{description}
\item[Case 1: $|\mathcal{C}| \le \frac nD$.]

$\forall (v,\phi(v))\in C$ we have $\phi(v) = 0$, since the core ends before column $\mathcal{B}$ by Definition \ref{def:bcore}, and so $C$ will be hit by the $(n,\frac nD)$-selective family $R$. It will therefore be hit by $\mathcal{S}$ on some column $j<|R|<\mathcal{B}=\mathcal{B}\lceil\frac{|C|}{r}\rceil$. Note that this case is the reason we require the ceiling function in the definition of a block synchronizer, but not in an upper block synchronizer.

\item[Case 2: $|\mathcal{C}|> \frac nD$.]
If $|\mathcal{C}| > \frac nD$, then it will be hit by a column $j<\frac{B \cdot |C|}{r}$ in the upper block synchronizer $S$, which corresponds to the column $j+\lceil\frac{j}{B}\rceil|R|$ in $\mathcal{S}$. Since $j+\lceil\frac{j}{B}\rceil|R| < \frac{B \cdot |C|}{r}+\lceil\frac{|C|}{r}\rceil|R|\le (B+R)\lceil\frac{|C|}{r}\rceil = \mathcal{B}\lceil\frac{|C|}{r}\rceil$, this satisfies the block synchronizer property.
\end{description}

So, $\mathcal{S}$ hits all cores $\mathcal{C}$ with $|\mathcal{C}|<\maxdeg$ within $\mathcal{B}\lceil\frac{|C|}{r}\rceil$ columns, and therefore hits all sets $X$ within $\mathcal{B}\lceil\frac{|X|}{r}\rceil$ under any activation schedule, fulfilling the criteria of an $(n, \maxdeg, \frac{n}{D}, \hat{c} \frac{n}{D}\log D \log\log \frac{D \maxdeg}{n})$-block synchronizer.
\end{proof}


\section{Conclusions}

The task of broadcasting in radio networks is a longstanding, fundamental problem in communication networks. Our result for deterministic broadcasting in directed networks combines elements from several of the previous works with some new techniques, and, in doing so, makes a significant improvement to the fastest known running time. Our algorithm for wake-up also improves over the previous best running time, in both directed and undirected networks, and relies on a proof of smaller universal synchronizers, a combinatorial object first defined in \cite{-CK04}.

Neither of these algorithms are known to be optimal. The best known lower bound for both broadcasting and wake-up is $\Omega(\min(n \log\ecc,\ecc \maxdeg \log \frac n\maxdeg))$ \cite{-CMS03}; our broadcasting algorithm therefore comes within a log-logarithmic factor, but our wake-up algorithm remains a logarithmic factor away.

As well as the obvious problems of closing these gaps, there are several other open questions regarding deterministic broadcasting in radio networks. Firstly, the lower bound for undirected networks is weaker than that for directed networks \cite{-KP03b}, and so one avenue of research would be to find an $\Omega(n \log \ecc)$ lower bound in undirected networks, matching the broadcasting time of \cite{-K05}. Secondly, the algorithms given here, along with almost all previous work, are non-explicit, and therefore it remains an important challenge to develop explicit algorithms that can come close to the existential upper bound. The best constructive algorithm known to date is by \cite{-I02}, but it is a long way from optimality.

Some variants of the model also merit interest, in particular the model with collision detection. It is unknown whether the capacity for collision detection improves deterministic broadcast time, as it does for randomized algorithms \cite{-GHK13}. Collision detection does remove the requirement of spontaneous transmissions for the use of the $O(n)$ algorithm of \cite{-CGGPR00}, but a synchronized global clock would still be required. It should be noted that collision detection renders the wake-up problem trivial, since if every active node transmits in every time-step, collisions will wake up the entire network within $\ecc$ time-steps.

\bibliographystyle{plain}
\bibliography{BC}
\end{document}